\documentclass[runningheads]{llncs}
\usepackage[numbers,sort&compress]{natbib}
\usepackage{graphicx}
\usepackage{multirow,enumitem}
\usepackage{comment,color}
\usepackage{amssymb}    
\usepackage{amsmath}
\usepackage{algorithm}
\usepackage{algorithmic}
\usepackage{appendix}

\usepackage{amsthm}
\usepackage{float}

\newtheorem{insight}{Insight}

\newcommand{\ignore}[1]{}
\makeatletter
\let\bfseries=\undefined
\DeclareRobustCommand\bfseries
{\not@math@alphabet\bfseries\mathbf
	\boldmath\fontseries\bfdefault\selectfont}
\makeatother

\def\Orthant_j{{\mathcal O}_{j}}

\usepackage{mathtools}
\makeatletter
\DeclareRobustCommand\widecheck[1]{{\mathpalette\@widecheck{#1}}}
\def\@widecheck#1#2{%
    \setbox\z@\hbox{\m@th$#1#2$}%
    \setbox\tw@\hbox{\m@th$#1%
       \widehat{%
          \vrule\@width\z@\@height\ht\z@
          \vrule\@height\z@\@width\wd\z@}$}%
    \dp\tw@-\ht\z@
    \@tempdima\ht\z@ \advance\@tempdima2\ht\tw@ \divide\@tempdima\thr@@
    \setbox\tw@\hbox{%
       \raise\@tempdima\hbox{\scalebox{1}[-1]{\lower\@tempdima\box
\tw@}}}%
    {\ooalign{\box\tw@ \cr \box\z@}}}
\makeatother

\begin{document}
\title{Optimally Blending Honeypots into Production Networks: Hardness and Algorithms}
%
%

\author{Md Mahabub Uz Zaman\inst{1}
\and
Liangde Tao\inst{2}
\and 
Mark Maldonado\inst{3}
\and
Chang Liu\inst{1}
\and
Ahmed Sunny\inst{1}
\and
Shouhuai Xu\inst{3}
\and
Lin Chen\inst{1}
}
%
%

\institute{{Texas Tech University, Lubbock TX 79409, USA} \and
{Zhejiang University, Hangzhou 310027, China} \and
{University of Colorado Colorado Springs, Colorado Springs, CO 80918, USA}}

\maketitle       

\begin{abstract}
\vspace{-6mm}
Honeypot is an important cyber defense technique that can expose attackers' new attacks (e.g., zero-day exploits). However, the effectiveness of honeypots has not been systematically investigated, beyond the rule of thumb that their effectiveness depends on how they are deployed. In this paper, we initiate a systematic study on characterizing the cybersecurity effectiveness of a new paradigm of deploying honeypots: blending honeypot computers (or IP addresses) into production computers. This leads to the following Honeypot Deployment (HD) problem: {\em How should the defender blend honeypot computers into production computers to maximize the utility in forcing attackers to expose their new attacks while minimizing the loss to the defender in terms of the digital assets stored in the compromised production computers?} We formalize HD as a combinatorial optimization problem, prove its NP-hardness, provide a near-optimal algorithm (i.e., polynomial-time approximation scheme). We also conduct simulations to show the impact of attacker capabilities.

\end{abstract}

\section{Introduction}

Cyberspace is complex and extremely challenging to defend because there are so many vulnerabilities that can be exploited to compromise its components, including both technological ones (e.g., software or network configuration vulnerabilities) and non-technological ones (e.g., human factors) \cite{Pendleton:2016,RosaSocialEngineeringKillChainSciSec2022}.  It would be ideal if we could prevent all attacks; unfortunately, this is not possible for reasons that include {\em undecidability} of computer malware \cite{Adleman-Crypto-88}. Not surprisingly, cyber attacks have caused tremendous damages \cite{Solarwinds2021,Cybersecurity-Venture}. 

\ignore{
To defend cyberspace against attacks, there is a range of approaches, including: preventive cyber defenses (e.g., access control), which aim to prevent attacks as much as possible; reactive cyber defenses (e.g., anti-malware tools and intrusion detection systems), which aim to detect successful attacks; adaptive cyber defenses, which aim to dynamically adapt the cyber defense posture after detecting attack signals (e.g., strengthening cybersecurity policy on demand); proactive cyber defenses, which aim to dynamically adapt the cyber defense posture despite that no attack signals have been detected (e.g., moving-target defense or honeypot); active cyber defenses, which aim to use 
``good-ware'' or ``green-ware,'' to detect and clean up compromises and ``kill'' malware \cite{XuBookChapterCD2019}. 
}

Honeypot \cite{provos2004virtual,cohen2006use,nawrocki2016survey} is a deception technique for luring and exposing cyber attacks, especially new attacks or zero-day exploits. The basic idea is to set up fake services that are open in the Internet, meaning that any access to these fake services can be deemed as malicious and the defender can learn the attacks by monitoring these fake services. The importance and potential of honeypots have attracted a due amount of attention (e.g.,
\cite{wagener2011adaptive,DBLP:conf/gamesec/HuangZ20,DBLP:conf/gamesec/HuangZ19,new.att.honeypot,honeypot.pca,li2010towards,DBLP:journals/ccr/KreibichC04,DBLP:journals/cn/PortokalidisB07,DBLP:conf/uss/AnagnostakisSAXMK05}). However, the effectiveness of honeypots has not been systematically characterized. The rule of thumb is that their effectiveness depends on how they are deployed. The traditional way of deploying honeypots is to isolate a set of honeypot computers from any production network. However, it is well known that such honeypots can be easily figured out, and thus evaded, by attackers. One approach to addressing this problem is to ``blend'' honeypot computers into the production computers of an enterprise network.

\noindent{\bf Our Contributions}. This paper makes two contributions. The {\em conceptual} contribution is to initiate the study on systematically characterizing the effectiveness of blending honeypot computers with production computers, leading to the formalization of a Honeypot Deployment (HD) problem: {\em How should the defender blend honeypot computers with production computers to maximize the utility of honeypot in forcing attackers to expose their new attacks,
while minimizing the loss to the defender in terms of the digital assets stored in the compromised production computers?} 
One salient feature of the formalization is that it can naturally incorporate attacker's {\em risk attitude} (i.e., risk-seeking, risk-neutral, or risk-averse).
The {\em technical} contribution is that we show: (i) the decision version of the HD problem is NP-complete; and (ii) we present a near-optimal algorithm to solve it, namely a Polynomial-Time Approximation Scheme (PTAS) by 
leveraging a given sequence of attacker’s preference (i.e., attack priority) resulting from the attacker's reconnaissance process and the attacker's risk attitude.
We also conduct simulation studies to draw insights into the aspects which we cannot analytically treat yet, which would shed light on future analytic research.


\section{Problem Statement}
\label{sec:problem-statement} 
\noindent{\bf Intuition}. It is non-trivial to model the Honeypot Deployment (HD) problem, 
so we start with an intuitive discussion.
Consider (for example) an enterprise network with some {\em production} computers, which provide real business services, and a set of IP addresses. Some IP addresses are assigned to these production computers. The defender deploys some {\em traditional} defense tools (e.g., anti-malware tools and intrusion-prevention systems) to detect and block recognizable attacks. However, these tools can be evaded by new attacks (e.g., zero-day exploits), which are not recognizable by them, per definition. In order to defend the network against new attacks, the defender can blend some {\em honeypot} computers into the production ones, meaning that some of the remaining (or unassigned) IP addresses are assigned to honeypot computers, and some IP addresses may not be used at all (in which case we say these IP addresses are assigned to {\em dummy} computers). Each computer (production, honeypot, and dummy alike) will be assigned one unique IP address. Note that traditional defense tools are still useful because they can detect and block recognizable attacks.

The research is to investigate how to {\em optimally} assign IP addresses to computers to benefit the defender, where the meaning of optimization is specified as follows.
{\em First}, suppose deploying one honeypot computer incurs a cost to the defender, which is plausible because the honeypot computer does not provide any business-related service. 
{\em Second}, when a new attack is waged against a production computer, it incurs a loss to the defender because the digital assets stored in the production computer are compromised and the new attack cannot be blocked by traditional defense tools. 
{\em Third}, when a new attack is waged against a honeypot computer, it incurs no loss to the defender but does incur a cost to the attacker because the new attack now becomes recognizable to the defender.
In this case, we say a {\em valid} new attack becomes {\em invalid} (i.e., no more useful to the attacker).
{\em Fourth}, the usefulness of honeypot computers is based on the premise that the attacker does not know which IP addresses are assigned to honeypot computers; otherwise, the attacker can simply avoid attacking them. In the real world, the attacker often uses a {\em reconnaissance} process, which can be based on a range of techniques (from social engineering to technical methods), to help determine which IP addresses may be assigned to honeypot computers. The reconnaissance process often correctly detects which IP addresses are assigned to the dummy computers because no attempt is made by the defender to disguise these IP addresses (otherwise, they can be deemed as honeypot computers). The reconnaissance process is not perfect in identifying the honeypot computers, meaning that when the attacker decides to attack a computer, which the attacker deems as a production computer, the computer is actually a honeypot one, causing the attacker to lose the new attack. Because of the uncertainty associated with the outcome of the reconnaissance process, the attacker would decide whether to attack a computer with some probability, which reflects the attacker's reconnaissance capability, 
the attacker's risk attitude (i.e., risk-seeking, risk-neutral, or risk-averse), and the honeypot computer's capability in disguising itself as a production computer.
To accommodate attacker reconnaissance capabilities, we assume the probabilities are given; this is reasonable because deriving such probabilities 
is orthogonal to the focus of this study. 


\noindent{\bf Problem Formalization}.
Let $\mathbb{N}$ denote the set of positive integers and $\mathbb{R}$ the set of real numbers. For any positive integer $z\in \mathbb{N}$, we define $[z]=\{1,\ldots,z\}$. 
Suppose the defender is given $n\in\mathbb{N}$ production computers and $n+m$ IP addresses, meaning that $m\in\mathbb{N}$ is the number of IP addresses that can be used to deploy honeypot computers and dummy computers (if applicable). 
{Suppose the defender needs to deploy up-to $m\in\mathbb{N}$ honeypots computers. Each honeypot computer may incur a cost $c\in\mathbb{N}$, which may vary depending on the degree of sophistication embedded into the honeypot computer (i.e., the most sophisticated a honeypot computer, the more difficult for the attacker's reconnaissance to determine whether it is a honeypot or production computer). Suppose the total budget for deploying the honeypot computers is $B\in \mathbb{N}$. 
This means that the defender will select $h$-out-of-the-$m$ IP addresses for deploying honeypot computers subject to the total cost for deploying the $h$ honeypot computers is at most $B$. 
Then, $m-h$ dummy computers are respectively deployed at the remaining $m-h$ IP addresses.}
Recall that the attacker can correctly recognize the dummy computers and will not attack them. We use the term ``non-dummy computer'' to indicate a computer that is a production or honeypot computer.

For ease of reference, we use ``computer $j$'' to denote ``the computer assigned with IP address $j$, where the computer may be a production, honeypot, or dummy one.'' Let $v_{j,D}$ be the value of the digital assets stored in computer $j$ (e.g., sensitive data and/or credentials).
This leads to a unified representation: an IP address $j\in [n+m]$ is associated with a value $v_{j,D}\in\mathbb{N}$ as assessed by the defender and a cost $c_j\in\mathbb{N}$ to the defender, where 
\begin{itemize}[noitemsep,nolistsep]
    \item $v_{j,D}>0$ if $j$ is a production computer, and $v_{j,D}=0$ otherwise.
    \item $c_j>0$ if $j$ is a honeypot computer, and $c_j=0$ otherwise.
\end{itemize}

Suppose the attacker has $r$ valid new attacks which are not recognized by the traditional defense tools employed by the defender, where $r$ represents the attacker's budget. Before using these attacks, the attacker often conducts a reconnaissance process to identify: (i) the value of digital assets stored in computer $j\in [n+m]$, denoted by $v_{j,A}\in\mathbb{N}$, which is the attacker's perception of the ground-truth value $v_{j,D}$ that is not known to the attacker (otherwise, the attacker would already know which computers are honeypot ones); and (ii) the probability or likelihood that a non-dummy computer $j$ is a honeypot computer, denoted by $q_j$, which is the probability that the attacker will {\em not} attack computer $j$ (i.e., $1-q_j$ is the probability that the attacker will attack it). 
{ The attacker knows which computer is a dummy one and will not attack any dummy computers. To summarize, let $x_j\in\{0,1\}$ be an indicating vector such that $x_j=1$ if computer $j$ is a production or honeypot computer, and $x_j=0$ if computer $j$ is a dummy computer, then we can see that the attacker will attack computer $j$ with probability $(1-q_j)x_j$, which is $1-q_j$ for a non-dummy computer $j$ and $0$ otherwise.} Note that $v_{j,A}$ and $q_j$ together reflect the attacker's reconnaissance capability.

When the attacker attacks computer $j$ which happens to be a production one, the loss to the defender (i.e., the reward to the attacker) is $v_{j,D}>0$; 
when the attacker attacks computer $j$ which happens to be a honeypot one, the loss to the defender is $v_{j,D}=0$ and the attacker's budget decreases by 1 because the new attack now becomes {\em invalid} (i.e., recognizable to the defender).
The attacker cannot wage any successful attack after its $r$ new attacks become invalid.

The research question is to identify the optimal strategy in assigning some of the $m$ IP addresses to { honeypot computers} under budget constraint $B$ so as to minimize the {\em expected loss} to the defender.
{ The assignment of IP addresses to the production, honeypot, and dummy computers is called a {\em defense solution}, which is characterized by a vector $\vec{x}=(x_1,x_2,\cdots,x_{m+n})\in\{0,1\}^{n+m}$ where $x_j=1$ means computer $j$ is a production or honeypot computer, and $x_j=0$ means it is a dummy computer.}
Given a fixed {\em defense solution} $\vec{x}$, 
the loss to the defender is defined as the total value of the production computers that are attacked. The loss to the defender is probabilistic because the attacker attacks a non-dummy computer with a probability, meaning that we should consider the {\em expected loss}. 
To compute the expected loss, we need to specify the probability distribution of the loss, which depends on the probability distribution of the attacker's decisions on attacking non-dummy computers. 
To characterize this distribution, we introduce two concepts, {\em attack sequence} and {\em attack scenario}; {the former is a stepping-stone for introducing the latter, which is used to compute the expected loss}. Given a defense solution, we assume the attacker sequentially decides whether to attack computer $j\in [n+m]$ according to probability $q_j$. The order according to which the attacker makes decisions is called {\em attack sequence}. 

\begin{figure}[!htbp]
	\centering
		\vspace{-.5em}
	\includegraphics[width=100mm]{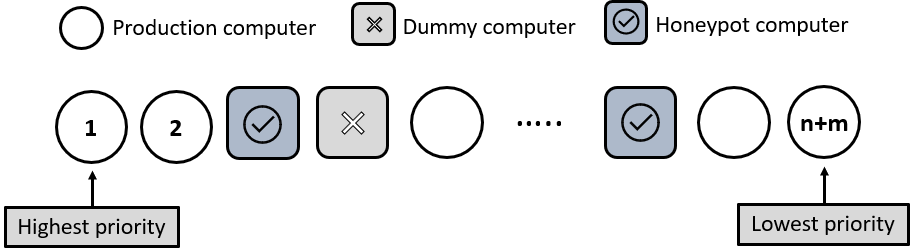}
	\vspace{-0.5em}
	\caption{Illustrating the concepts of production, honeypot, and dummy computers and the idea of {\em attack sequence}.}
	\label{fig:fig_0}
\end{figure}

{As illustrated in Figure \ref{fig:fig_0}, we use circles to represent production computers and squares to represent the $m$ IP addresses for which the defender needs to decide whether to deploy a honeypot or dummy computer.
An attack sequence represents the attacker's choice of priority in considering which non-dummy computers to attack, where priority depends on the attacker's perception of the value of computer $j$, namely $v_{j,A}$, and the probability $q_j$, and the attacker's risk attitude. We assume such attack sequences are given as input to the present study because attaining these attack sequences is an orthogonal research problem.}
To define {\em attack scenario}, we should specify when the attacker stops. The attacker stops when any of the following two conditions hold: (i) after attacking some honeypot computer which causes the attacker's budget to decrease from $1$ to 0, meaning that the attacker has no more valid new attack to use; (ii) the attacker finishes attacking all computers it would like to attack (based on its probabilistic decision) even if there are still valid new attacks, meaning that the attacker only considers whether or not to attack a non-dummy computer once, which is plausible. {Specifically, given a defense solution $\vec{x}$ and an attack sequence,} 
the decisions on whether or not to attack computers $j\in [j^*]$
is called an {\it attack scenario}; that is, an attack scenario $s$ is a binary vector $\pi_s=(\pi_s(1),\pi_s(2),\cdots,\pi_s(j_s))\in\{0,1\}^{j_s}$,
where $\pi_s(j)=1$ means computer $j$ is indeed attacked and $\pi_s(j)=0$ otherwise, and $j_s\le n+m$ is the last computer that is attacked. Note that $\pi_s(j_s)=1$ by definition. Recall that the attacker attacks computer $j$ with probability $1-q_j$, meaning $\Pr[\pi_s(j)=1]=1-q_j$ and $\Pr[\pi_s(j)=0]=q_j$. Define $\Pr[\pi_s]$ as the probability that attack scenario $\pi_s$ occurs, then we have
$$\Pr[\pi_s]{=}\prod_{j\in[j_s]:\pi_s(j)=1, x_j=1}(1-q_j)\cdot\prod_{j\in [j_s]:\pi_s(j)=0,x_j=1}q_j.$$
Note that if $x_j=0$ then computer $j$ is a dummy computer and the attacker will not attack it at all, so it does not contribute to the probability $\Pr[\pi_s]$. Let $\mathcal{P}\subset [n+m]$ be the subset of the IP addresses that are assigned to production computers and $\mathcal{H}:=[n+m]\setminus \mathcal{P}$ be the subset of the remaining IP addresses. 
If attack scenario $\pi_s$ occurs, then the loss to the defender is the total value of all production computers attacked by the attacker, which is
$Loss(\pi_s){=}\sum_{j\in\mathcal{H}:\pi_s(j)=1}v_{j,D}.$
Given the loss incurred by a specific attack scenario as shown in the above equation, 
the expected loss is defined over all possible attack scenarios. Not every vector in $\{0,1\}^{k}$ where $k\le n+m$ is necessarily an attack scenario; a vector $\pi\in\{0,1\}^k$ is an attack scenario if and only if
\begin{itemize}
\item $k=n+m$ and $|\{j\in [n+m]\cap \mathcal{H}:\pi(j)=1\}|\le r$; or
\item $k<n+m$, $\pi(k)=1$ and $|\{j\in [k]\cap\mathcal{H}:\pi(j)=1\}|= r$. 
\end{itemize}
We call a vector $\pi$ satisfying the preceding condition an {\em attack scenario-compatible} vector. Let $\mathcal{V}$ be the set of all attack scenario-compatible vectors. 
Then, the expected loss of the defender {with respect to a fixed defense solution} is:
\begin{equation}
\label{eq:expected-loss}\mathbb{E}[Loss]=\sum_{\pi\in \mathcal{V}} Loss(\pi)\Pr[\pi].
\end{equation}
In summary, we have:

\vspace{-1.5em}
		
\begin{center}
	\fbox{\begin{minipage}{\textwidth}	
{\bf Honeypot Deployment (HD) Problem}\\			
\textbf{Input}: There are $n\in\mathbb{N}$ production computers and {a budget of $B\in\mathbb{N}$ for the defender.}
			There are $n+m$ IP addresses, which are indexed as $1,\ldots,n+m$. 
			Among these $n+m$ IP addresses, the $n$ production computers are respectively deployed at $n$ pre-determined IP addresses;
			among the remaining $m$ IP addresses, the defender will select a subset of them to deploy honeypot computers, and the other IP addresses will be assigned to dummy computers, which are known to, and not be attacked by, the attacker.
			For each computer $j\in [n+m]$, 
			there is an associated value $v_{j,D}\in \mathbb{N}$, 
			where $v_{j,D}=0$ if $j$ is a honeypot or dummy computer and $v_{j,D}>0$ otherwise; moreover, there is an associated 
			cost $c_j\in\mathbb{N}$, where $c_j=0$ if $j$ is a production or dummy computer, and $c_j>0$ if $j$ is a honeypot computer. The total cost incurred by deploying honeypot computers cannot exceed budget $B$.
			The attacker has $r\in\mathbb{N}$ valid new attacks.
			For computer $j\in [n+m]$, the attacker has a perceived value $v_{j,A}$ and a probability $q_j$ that a non-dummy computer $j$ is a honeypot computer. The attacker needs to use one valid new attack to attack a {\em non-dummy} computer $j$. The attacker attacks a non-dummy computer with probability $1-q_j$, and does not attack a dummy computer. 
			If the attacker indeed attacks a non-dummy computer $j$, there are two cases: in the case $j$ is a honeypot computer, then the loss to the defender is $v_{j,D}=0$ and the attacker's budget decreases by 1; in the case $j$ is a production computer, then the loss to the defender is $v_{j,D}>0$ and the attack's budget remains unchanged.
			The attacker stops when its budget becomes $0$, meaning that all of its $r$ new attacks become invalid, or {it has made decisions on whether to attack the $n+m$ computers}. 

			\noindent \textbf{Output}: Decide which of the $m$ IP addresses should be assigned to honeypot computers within budget $B$ so as to
			minimize the expected loss defined in Eq.\eqref{eq:expected-loss}.
	\end{minipage}}
\end{center}

\section{Hardness and Algorithmic Results}
\label{sec:algorithmic-result}

In this section, we study the computational complexity of, and algorithms for solving the HD problem, assuming {an attack sequence} is given. 
Without loss of generality, we can index the computers so that the attack {sequence} 
is $(1,2,\cdots,n+m)$. 
As illustrated in Figure \ref{fig:fig_0}, we use circles to represent production computers and squares to represent the $m$ IP addresses for which the defender needs to decide whether to deploy a honeypot or dummy computer. Recall
$\mathcal{P}$ is exactly the set of indices of the circles and $\mathcal{H}$ is exactly that of the squares. 

Note that the $v_{j,A}$'s, $q_j$'s together with the attacker's risk attitude decide the priority of the non-dummy computers to the attacker, and thus the attack sequence. 
Once an attack sequence is fixed, the objective value (i.e., expected loss to the defender) only depends on $v_{j,D}$'s. Since our hardness and algorithmic results are based on a fixed attack sequence, our discussion throughout this section does not involve $v_{j,A}$'s. Hence, we let $v_j=v_{j,D}$ for simplifying notations.

\subsection{Hardness Result}
\label{sec:hardness-result}

Now we study the decision version of the HD problem: decide whether or not there exists an assignment of the $m$ IP addresses to honeypot computers with budget $B$ such that the expected loss is no larger than the given threshold $T$. 
\vspace{-.25em}
\begin{theorem}\label{thm:np-h}
The decision version of the HD problem is NP-complete.
\end{theorem}
\vspace{-.25em}
The proof of Theorem~\ref{thm:np-h} is deferred to Appendix~\ref{appsec:hardness}. Here we discuss the basic idea behind the proof. Membership in NP is straightforward. Towards the NP-hardness proof, we reduce from the Subset Product problem. The instance and the solution of the Subset Product problem are given as follows:

\vspace{-1em}
\begin{center}
	\fbox{\begin{minipage}{\textwidth}	
	        \textbf{Subset Product Problem}
	        
		    \textbf{Input:} $k\in\mathbb{N}$, $S\!=\!\{1,\dots,m\}$, $w\!=\!(w_1,\dots,w_m)\!\in\!\mathbb{N}^m$.
		    
		    \textbf{Output:} Is there $S'\!\subseteq\! S$ such that $\prod_{i\in S'}\!w_i\!=k$?
	\end{minipage}}
\end{center}

A key fact for the Subset Product problem (which is different from the Subset Sum problem) is that Subset Product is NP-hard even if each $a_i$ is bounded by $m^{O(1)}$. This means we leverage the following Lemma~\ref{subset_product} given by Yao \cite{yao1980new}.

\begin{lemma}[\cite{yao1980new}]
    \label{subset_product}
    Assuming $P\neq NP$, the Subset Product problem can not be solved in $(mw_{max}\log k)^{O(1)}$ time where $w_{max}=\max_iw_i$.
\end{lemma}

\subsection{Algorithmic results}\label{subsec:algo}
As a warm-up, we present an exact algorithm, Algorithm \ref{alg:1}, 
to brute forces the optimal solution in exponential time. 
We show this algorithm can be modified to obtain Algorithm \ref{alg:2} to find a near-optimal solution in polynomial time. 

\setcounter{algorithm}{0} 
\begin{algorithm}[h]
	\caption{Dynamic Programming for the HD problem}\label{alg:1} 
	\begin{algorithmic}[1]
	    \renewcommand{\algorithmicrequire}{\textbf{Input:}}
        \renewcommand{\algorithmicensure}{\textbf{Output:}}
        \REQUIRE {
                $I:$ the attack sequence of IP address\\
                $q_{t}:$ the probability that attacker does not attack computer $t$ 
                \\
                $v_t:$ computer's value assigned with IP address $t$\\
                $c_t:$ cost of deploying honeypot computer at IP address $t$\\
                $B:$ budget of the defender
                }
        \ENSURE The assignment of IP address to honeypot computers which minimizes the expected loss of the defender and satisfies the total deployment cost is no greater than $B$. 
		\STATE{$\widetilde{\mathcal{F}}_0=\{(0,1,0,\dots,0)\}$ }
		\FOR{$t=1$ to $n+m$}
		\STATE {$\widetilde{\mathcal{F}}_t=\emptyset$}
		\FORALL {$(t\!-\!1,p_0, p_1,\dots,p_r,e,b) \in \widetilde{\mathcal{F}}_{t-1}$}
		
		\IF{IP address $t$ is assigned to a production computer}
		\STATE {$\widetilde{\mathcal{F}}_t\leftarrow \widetilde{\mathcal{F}}_{t} \cup (t,p_0,p_1,\dots,p_r,e+q_tv_t\sum_{i=0}^{r-1}p_i,b)$}
		\ELSE
		\STATE {$\widetilde{\mathcal{F}}_t\leftarrow \widetilde{\mathcal{F}}_{t} \cup (t,p_0, p_1,\dots,p_r,e,b)$}
		\FOR{$k = 1$ to $r$}
		\STATE {$p_{k}^{'} = (1-q_t)\cdot p_{k-1}+q_t\cdot p_k$}
		\ENDFOR
		\STATE 
		{$p_{0}^{'} = q_tp_0$}
		\STATE {$\widetilde{\mathcal{F}}_h\leftarrow \widetilde{\mathcal{F}}_{t} \cup (t,p_{0}^{'},\dots,p_{r}^{'},e,b+c_t)$}
		\ENDIF
		
		\ENDFOR
	    \STATE {eliminate all the dominated states in $\widetilde{\mathcal{F}}_t$}
		\ENDFOR
		\STATE {return $\min\{e:(n+m,p_0,p_1,\dots,p_r,e,b)\in \widetilde{\mathcal{F}}_{n+m}$ and $b\leq B\}$}
	\end{algorithmic}
\end{algorithm}

\vspace{-1.5em}

\subsubsection{An exact algorithm via dynamic programming}\label{subsubsec:algo1}
Algorithm \ref{alg:1} essentially branches on whether or not to deploy a honeypot computer for every $t\in\mathcal{H}$, thus the total number of distinct dominated states (e.g., $\sum_t|\widetilde{\mathcal{F}}_t|$) is bounded by $2^{O(m)}$. Hence, its running time is $2^{O(m)}$.
Algorithm \ref{alg:1} serves two purposes: (i) it provides a method to recursively compute Eq.\eqref{eq:expected-loss}, while noting that the definition of Eq.\eqref{eq:expected-loss} involves an exponential number of attack scenarios that cannot be used directly; (ii) it can be combined with rounding techniques to give a polynomial-time approximation scheme, which is the main algorithmic result (Algorithm \ref{alg:2}).

We consider the following sub-problem: Let $t\in [n+m]$. Is it possible to deploy honeypot computers {at some IP addresses within  $[t]$ such that (i) the total cost equals $b$ which is a given constraint, (ii) the expected loss to the defender because of attacks against the production computers in $\mathcal{P}\cap [t]$ equals $e$, and (iii) the probability that the attacker attacks exactly $k$ honeypot computers within $[t]$ is $p_k$ for every $k=0,1,\cdots,r$?  
We denote the sub-problem by a $(r+4)$-tuple $(t,p_0,p_1,\dots,p_r,e,b)$.}


We define set $\mathcal{F}_t$ as: If the answer to the sub-problem $(t,p_0,p_1,\dots,p_r,e,b)$ is ``yes,'' then we call $(t,p_0,p_1,\dots,p_r,$ $e,b)$ as a stage-$t$ state and store it in $\mathcal{F}_t$. 
Note that $\mathcal{F}_t$ can be computed recursively as follows. Suppose we have computed $\mathcal{F}_{t-1}$. 
Each stage-$(t\!-\!1)$ state gives rise to stage-$t$ states as follows: 
\begin{itemize}
\item  If $t\in \mathcal{P}$, i.e., IP address $t$ is associated with a production computer, then the stage-$(t-1)$ state $(t\!-\!1,p_0,p_1,\dots,p_r,e,b)$ gives rise for one stage-$t$ state $(t,p_0,p_1,\dots,p_r,e',b)$ where $e'=e+(1-q_t)v_t\sum_{k=1}^{r-1}p_k$.  We explain this equation as follows. By the definition of a state, $(t-1,p_0,p_1,\dots,p_r,e,b)\in \mathcal{F}_{t-1}$ implies that $p_r$ is the probability that the attacker has attacked $r$ honeypot computers {within $[t-1]$,} 
and thus the attacker cannot attack anymore. If this event happens, we have $e'=e$; otherwise (with probability $\sum_{k=1}^{r-1}p_k=1-p_r$) the attacker is able to attack the production computer at {IP address} $t$ which has a value $v_t$, and the attacker attacks it with probability $1-q_t$. {This leads to}
$e'=ep_r+(e+v_t)(1-q_t)\sum_{k=1}^{r-1}p_k=e+(1-q_t)v_t\sum_{k=1}^{r-1}p_k$. 

\item If $t\in \mathcal{H}$, then we have two options, i.e., we either assign a honeypot computer or a dummy computer {to IP address $t$}. If we assign a dummy computer, then the stage-$(t-1)$ state $(t\!-\!1,p_0,p_1,\dots,p_r,e,b)$ gives rise to  stage-$t$ state $(t,p_0,p_1,$ $\dots,p_r,e,b)$; if we assign a honeypot computer, then $(t-1,p_0,p_1,\dots,p_r,e,b)$ gives rise to stage-$t$ state $(t,p'_0,p'_1,$ $\dots,p'_r,e,b+c_t)$ where $p_0'=p_0q_t$ and $p_k'=p_{k}q_t+p_{k-1}(1-q_t)$ for $1\le k\le r$.
\end{itemize}

\begin{definition} We say stage-$t$ state $(t,p_0,p_1,\dots,p_r,e,b)$ dominates $(t,p_0,p_1,$ $\dots,$ $p_r,e,b')$ if it holds that $b<b'$.
\end{definition}

Denote by $\widetilde{\mathcal{F}}_t \subseteq \mathcal{F}_t$ the set of all stage-$t$ states which are not dominated by {any of the} other stage-$t$ states. Let $x^*$ denote the optimal solution of the HD problem. Let $\mathcal{H}(x^*)$ be the set of IP addresses that are assigned to {all honeypot computers}. Consider $\mathcal{H}(x^*)\cap [t]$, which is the set of IP addresses in $[t]$ that are assigned to a honeypot computer in the optimal solution $x^*$.   
Denote by $b_t(x^*)$ the total cost of deploying honeypot computers in $\mathcal{H}(x^*)\cap [t]$. Denote by $p_{t,i}(x^*)$ the probability that the attacker has attacked exact $i$ honeypot computers within $\mathcal{H}(x^*)\cap [t]$. 
Denote by $e_t(x^*)$ the expected loss to the defender from computers in $[t]$. The following lemma demonstrates that the optimal solution $x^*$ can be determined from $\widetilde{\mathcal{F}}_{n+m}$; its proof 
is deferred to Appendix~\ref{appsec:proof-lemmadp}.

\begin{lemma}\label{lemma:dp}
For each optimal solution $x^*$ of the HD problem and each $t\in[1,n+m]$, there exists some stage-$t$ state $(t,p_{t,0}(x^*),\dots,$ $p_{t,r}(x^*),e_t(x^*),b)\in \widetilde{\mathcal{F}}_{t}$ such that $b\leq b_t(x^*)$.
\end{lemma}



\vspace{-1.25em}

\subsubsection{A Polynomial-Time Approximation Scheme (PTAS)}\label{subsec:ptas}
Now we design a PTAS ({i.e., Algorithm \ref{alg:2}}) for the HD problem by modifying Algorithm \ref{alg:1}. The key idea is to reduce the total number of states that need to be stored during the dynamic programming. 
{We start with a high-level description of Algorithm \ref{alg:2}.
Let $\xi=\epsilon/2(n+m)$. Define $\Gamma_\xi=\{[0],(0,1],(1,1+\xi],\dots,((1+\xi)^{\gamma-1},(1+\xi)^\gamma]\}$ where $(1+\xi)^{\gamma-1}<nv_{\max}\leq(1+\xi)^{\gamma}$. Define $\Lambda_\xi=\{[0],(0,(1+\xi)^{-\gamma}],((1+\xi)^{-\gamma},(1+\xi)^{-\gamma+1}],\dots,((1+\xi)^{-1},1]\}$. The high dimensional area $[0,1]^{r+1}\times [0,(1+\xi)^{\gamma}]$ is then divided into a collection of boxes where each box $\mathcal{I}\in \Lambda_\xi^{r+1} \times\Gamma_\xi$. 
In each box, only one representative state will be constructed and stored in $\widehat{\mathcal{F}}_t$. $\widehat{\mathcal{F}}_t$ is computed recursively in two steps: (i). Given $\widehat{\mathcal{F}}_{t-1}$, each of its state gives rise to stage-$t$ states following the same formula as Algorithm \ref{alg:1} (see line 5 to line 13 of Algorithm \ref{alg:2}). Here $\mathcal{S}_t$ is introduced as a temporary set that contains all the stage-$t$ states computed from $\widehat{\mathcal{F}}_{t-1}$. (ii). Within each box $\mathcal{I}$, if $\mathcal{S}_t$ contains multiple states, then only the state with the minimal value in coordinate $b$ will be kept. All other states are removed. By doing so we obtain $\widehat{\mathcal{F}}_t$ from $\mathcal{S}_t$.} {Details of Algorithm \ref{alg:2} are presented below.} 
The rest of this subsection is devoted to proving the following theorem.
\begin{theorem}\label{thm:al2}
Algorithm \ref{alg:2} gives an $(1+\epsilon)$-approximation solution for the HD problem and runs in $(\frac{n+m}{\epsilon})^{O(r)}\log(v_{\max})$ time where $v_{\max}=\max_iv_i $.
\end{theorem}

To prove Theorem~\ref{thm:al2}, we need the following lemma that estimates the error accumulated in the recursive calculation of Algorithm \ref{alg:2} and illustrates the relationship between $\widetilde{\mathcal{F}}_t$ and $\widehat{\mathcal{F}}_t$.
\begin{lemma}\label{lemma:round}
For each $(t,p_0,p_1,\dots,p_r,e,b)\in \widetilde{\mathcal{F}}_t$, there exists $(t,\hat{p}_0,\hat{p}_1,\dots,\hat{p}_r,\hat{e},\hat{b})\in \widehat{\mathcal{F}}_t$ such that $\hat{b}\leq b$, $(1-\xi)^t\hat{e} \leq e \leq \hat{e}$, and for $i=0,\dots,r$ it holds that $(1-\xi)^t\hat{p}_i \leq p_i \leq \hat{p}_i$.
\end{lemma}

\vspace{-2em}

\setcounter{algorithm}{1} 
\begin{algorithm}[!h]
	\caption{Improved Dynamic Programming for the HD problem}
	\label{alg:2}
	\begin{algorithmic}[1]
	\renewcommand{\algorithmicrequire}{\textbf{Input:}}
        \renewcommand{\algorithmicensure}{\textbf{Output:}}
        \REQUIRE {$I,q_{t},v_t,c_t,B$
                }
        \ENSURE The assignment of IP address to honeypot computers which minimizes the expected loss of the defender and satisfies the total deployment cost is no greater than $B$. 
		
        \STATE{$\widehat{\mathcal{F}}_0=\{(0,\dots,0)\}$ }
		\FOR{$t=1$ to $n+m$}
		\STATE {${\mathcal{S}}_t=\emptyset$}
		\FORALL {$(t\!-\!1,p_1,\dots,p_r,e,b) \in \widehat{\mathcal{F}}_{t-1}$}
		
		\IF{IP address $t$ is assigned to a production computer}
		\STATE {${\mathcal{S}}_t\leftarrow {\mathcal{S}}_{t} \cup (t,p_0,p_1,\dots,p_r,e+q_tv_t\sum_{i=0}^{r-1}p_i,b)$}
		\ELSE
		\STATE {${\mathcal{S}}_t\leftarrow {\mathcal{S}}_{t} \cup (t,p_0, p_1,\dots,p_r,e,b)$}
		\FOR{$k = 1$ to $r$}
		\STATE {$p_{k}^{'} =(1- q_t) p_{k-1}+q_t\cdot p_k$}
		\ENDFOR
		\STATE {${\mathcal{S}}_t\leftarrow {\mathcal{S}}_{t} \cup (t,q_tp_0,p_{1}^{'},\dots,p_{r}^{'},e,b+c_t)$}
		\ENDIF
		\ENDFOR
		\STATE{$\widehat{\mathcal{F}}_t=\emptyset$}
		\FORALL {Box $\mathcal{I}\in (\Lambda_\xi)^{r+1}\times\Gamma_\xi$}
		\FOR{$i=0$ to $r$}
	    \STATE {$\hat{p}_i=\max\{{p}_i:(h,p_0,p_1,\dots,{p}_i,\dots,p_r,e,b)\in {\mathcal{S}}_t\cap \mathcal{I}\} $}
		\ENDFOR
		\STATE {$\hat{e}=\max\{{e}:(h,p_0,p_1,\dots,p_r,{e},b)\in {\mathcal{S}}_t\cap \mathcal{I}\} $}
		\STATE {$\hat{b}=\min\{{b}:(h,p_0,p_1,\dots,p_r,e,{b})\in {\mathcal{S}}_t\cap \mathcal{I}\} $} \label{sec:alg-b}
	    \STATE {$\widehat{\mathcal{F}}_t \leftarrow \widehat{\mathcal{F}}_t  \cup (t,\hat{p}_0,\hat{p}_1,\dots,\hat{p}_r,\hat{e},\hat{b})$}
		\ENDFOR
		\ENDFOR
		\STATE {return $\min\{\hat{e}:(n,\hat{p}_0,\hat{p}_1,\dots,\hat{p}_r,\hat{e},\hat{b})\in \widehat{\mathcal{F}}_{n+m}$ and $\hat{b}\leq B\}$}
	\end{algorithmic}
\end{algorithm}

\vspace{-2em}

\begin{proof} 
We prove this by induction. Clearly, Lemma~\ref{lemma:round} holds for $t=1$. Suppose it holds for $t=\ell-1$, i.e., for each $(\ell-1,p_0,p_1,\dots,p_r,e,b)\in \widetilde{\mathcal{F}}_{\ell-1}$, there exists $(\ell-1,\hat{p}_0,\hat{p}_1,\dots,\hat{p}_r,\hat{e},\hat{b})\in \widehat{\mathcal{F}}_{\ell-1}$ such that $\hat{b}\leq b$, $(1-\xi)^{k-1}\hat{e} \leq e \leq \hat{e}$, and for $i=0,\dots,r$ it holds that $(1-\xi)^{k-1}\hat{p}_i \leq p_i \leq \hat{p}_i$. We prove Lemma~\ref{lemma:round} for $t=\ell$.

Note that the recursive computation is the same as Algorithm \ref{alg:1}, the only difference is that we replace the accurate value $p_k$'s with the approximate value $\hat{p}_k$'s. The rounding error will accumulate through the calculation, but will not increase too much in each step through the following two observations: (i) For any $\alpha\in[0,1]$ and any $j\in[1,r]$ it holds that $$(1-\xi)^{\ell-1}(\alpha\hat{p}_{j-1}+(1-\alpha)\hat{p}_j) \leq \alpha{p}_{j-1}+(1-\alpha){p}_j  \leq  \alpha\hat{p}_{j-1}+(1-\alpha)\hat{p}_j.$$
(ii) For any $\beta\in \mathbb{R}_{\geq 0}$ it holds that $$(1-\xi)^{\ell-1}[\hat{e}+\beta\sum_{i=1}^{r-1}\hat{p}_i] \leq e+\beta\sum_{i=1}^{r-1}p_i \leq \hat{e}+\beta\sum_{i=1}^{r-1}\hat{p}_i.$$

Hence, we know that for each $(\ell,p_0,p_1,\dots,p_r,e,b)\in \widetilde{\mathcal{F}}_\ell$, there exists $(\ell,p'_0,p'_1$ $,\dots,p'_r,e',b')\in {\mathcal{S}}_\ell$ such that $b'\leq b$, $(1-\xi)^{\ell-1}e' \leq e \leq e'$, and for $i=0,\dots,r$ it holds that $(1-\xi)^{\ell-1}p'_i \leq p_i \leq p'_i$. 

We know that within each box $\mathcal{I}$, only one representative state $(t,\hat{p}_0,\dots,\hat{p}_r,\hat{e})$ will be constructed and stored in $\hat{\mathcal{F}}_t$. By the definition of $\Lambda_\xi$ and $\Gamma_\xi$, we know that for each $(k,{p}'_0,{p}'_1,\dots,{p}'_r,{e}',{b}')\in {\mathcal{S}}_t\cap \Gamma$ there exists $(t,\hat{p}_0,\hat{p}_1,\dots,\hat{p}_r,\hat{e},\hat{b})\in \hat{\mathcal{F}}_t$ such that $\hat{b}\leq b'$, $(1-\xi)\hat{e} \leq e' \leq \hat{e}$, and for $i=0,\dots,r$ it holds that $(1-\xi)\hat{p}_i \leq p'_i \leq \hat{p}_i$. Thus, Lemma~\ref{lemma:round} is proved.
\end{proof}

Now we are ready to prove Theorem~\ref{thm:al2}.
\begin{proof}[Proof of Theorem~\ref{thm:al2}]
We first estimate the overall error incurred in Algorithm \ref{alg:2}. Let $\xi=\epsilon/2(n+m)$. Since $1-(n+m)\xi\leq (1-\xi)^{n+m}$, it is easy to verify that $(1-\xi)^{-n-m}\leq 1+\epsilon$. According to Lemma~\ref{lemma:round}, Algorithm \ref{alg:2} gives an $(1+\epsilon)$-approximation solution for the HD problem, i.e., the expected loss of the defender is no larger than $(1+\epsilon)$ times the minimum expected loss of the defender. 

Now we estimate the overall running time. The total number of distinct dominated states (e.g., $\sum_t|\widehat{\mathcal{F}}_t|$) is bounded by $O((n+m)|\Lambda_\xi|^{r+1}|\Gamma_\xi|)$. We know that $|\Lambda_\xi|\leq O(\frac{n+m}{\epsilon})$. In the meantime, $|\Gamma_\xi|=\log_{1+\xi}(nv_{\max})\leq O(\frac{n+m}{\epsilon}\log(v_{\max}))$ where $v_{\max}=\max_iv_i$. Overall, Algorithm \ref{alg:2} runs in $(n+m)^{r+3}\log(v_{\max})/\epsilon^{r+2}$ time. 
Hence, Theorem~\ref{thm:al2} is proved.
\end{proof}

\section{Experiment}
\label{sec:expreimental-result}




\noindent{\bf Simulation Parameters}.
In our simulation, we set the number of production computers as $n= 255$, 
$m\in \{15,20,25,30\}$, the number of attacker's new attacks as $r\in\{5,10,15\}$, and the defender's budget as $B=\{1000,2000,3000,4000\}$. The defender's perceived value for each production computer $j$, $v_{j,D}$, is generated uniformly at random within $[50,2000]$ (while noting that $v_{j,D}=0$ if $j$ is no production computer); the attacker's perceived value for each non-dummy computer $j$, $v_{j,A}$, is also generated uniformly at random within $[50,2000]$. The cost for deploying honeypot computer $j$ is generated uniformly at random within $[50,200]$. The probability $q_j \in [0,1]$ that the attacker believes $j$ is a honeypot computer will be set in specific experiments. 
{To conduct a fair comparison between the expected losses incurred in different parameter settings, we normalize them via {\em expected relative loss}, which is the ratio between the expected loss and the summation of all $v_{j,D}$, leading to a normalized range $[0,1]$.}

\subsection{Expected Losses under Different Attack Sequences}{\label{seq_gen}}
Now we study how different attack sequences may affect the optimal objective value of the HD problem. Recall that attack sequence is an input to our algorithms.
Given an attack sequence, we can apply Algorithm \ref{alg:2} to compute a near-optimal defense solution to the defender, which leads to essentially the smallest expected loss the defender can hope for. Consequently, the smallest expected loss, i.e., the optimal objective value, reflects how destructive the attacker is when it chooses a certain attack sequence.

While the attack sequence can be arbitrary, we are interested in the attack sequences that are likely to be adopted by a rational attacker. Note that the attacker observes $q_j$ and $v_{j,A}$ for each computer $j\in [n+m]$. Intuitively, the attacker needs to weigh between the potential gain $v_{j,A}$ and the risk that the attacker cannot get this gain, namely probability $q_j$. Therefore, the attacker's risk attitude determines the attack sequence.
Leveraging ideas from economics, we study three types of risk attitudes of the attacker (see, e.g.~\cite{hillson2007understanding, galinkin2021evaluating}): 
    (i) {\em risk-seeking}, meaning that the attacker wants to maximize its revenue fast;
    (ii) {\em risk-averse}, meaning that the attacker wants to minimize its chance of losing a valid new attack;
    (iii) {\em risk-neutral}, meaning that the attacker acts in between risk-seeking and risk-averse.
A formal definition of risk-attitude depends on the notion of {\em utility function}, denoted by $u_j$. We focus on a broad class of utility functions, known as exponential utility \cite{pratt1978risk,camerer2004advances}, which is defined as:

\[
    u_j= 
\begin{cases}
    \frac{1 - e^{-\alpha v_{j,A}} }{\alpha},& \text{if } 
    \alpha \neq 0\\
    v_{j,A},              & \text{if }  \alpha = 0
\end{cases}
\]
where $\alpha$ is the {\em coefficient of absolute risk aversion}, which, roughly speaking, measures how much the attacker is willing to sacrifice the expected value $v_{j,A}$ in order to achieve perfect certainty about the value it can receive. If $\alpha>0$, then the attacker is risk-aversion; if $\alpha=0$, the attacker is risk-neutral; if $\alpha<0$, the attacker is risk-seeking.  
With the utility function, the attacker (with risk attitude specified by $\alpha$) will rank (or prioritize) the non-dummy computers based on the non-increasing order of the expected utility value $(1 - q_j) \cdot u_j$ and use these tanks to formulate an attack sequence. 

\begin{figure}[!htbp]
\centering
	\includegraphics[width=.7\textwidth]{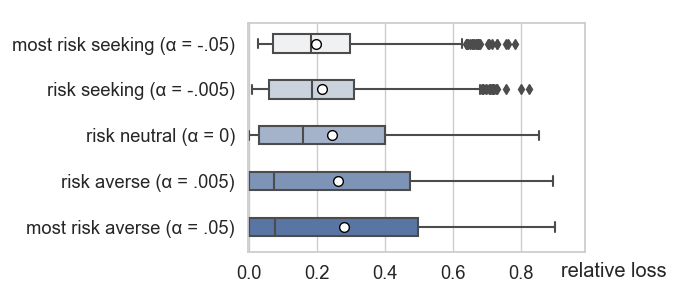}
 \vspace{-2em}
	\caption{Expected relative loss with respect to different risk-attitude {(i.e., $\alpha=-0.05$ for most risk-seeking, $\alpha=-0.005$ for risk-seeking, $\alpha=0$ for risk-neutral, $\alpha=0.005$ for risk-averse, $\alpha=0.05$ for most risk-averse)}.
	\label{fig:boxplot}}
\end{figure}

In our experiment, we choose $\alpha \in \{-0.05,-0.005,0,0.005,0.05\}$, where $\alpha = -0.05$ means the attacker is strongly risk-seeking and $\alpha = 0.05$ means the attacker is strongly risk-averse. For each $\alpha$, we generate 4,560 instances. Figure~\ref{fig:boxplot} uses the standard box-plot to summarize the expected relative loss with respect to different values of $\alpha$.
Recall that for box-plot, the left and right boundary of the rectangle respectively corresponds to the 25th and 75th percentile; the line in the middle marks the 50th percentile or median; the small empty circle within each box is the mean value; the black dots are outliers. 


\begin{insight}
\label{insight-1}
Under exponential utility, the expected relative loss with respect to a risk-seeking attacker 
has a smaller variance than that of a risk-averse attacker, meaning that defending against a risk-seeking attacker is more predictable.
\end{insight}

Insight \ref{insight-1} is counter-intuitive at first glance because risk-averse attackers are, by definition, more deterministic or prefer less variance.
However, it can be understood as follows: risk-averse attackers are very sensitive to the $q_j$'s. Among the randomly generated instances, we observe that the attack sequences of a risk-averse attacker can vary substantially for two instances with similar $q_j$'s; by contrast, the attack sequence of a risk-seeking attacker does not.
Since different attack sequences can cause significant changes to the expected relative loss, the expected relative loss of a risk-seeking attacker has a smaller variance in general.

\subsection{Expected Loss w.r.t. Attacker's Reconnaissance Capability}{\label{strength_of_atkr}}

An attacker's reconnaissance capability is reflected by the $q_j$'s and $v_{j,A}$'s. For a perfectly capable attacker,
it holds that $v_{j,A}=v_{j,D}$ (i.e., the attacker can correctly obtain the value of the digital assets in computer $j$),
$q_j=0$ for each production computer $j$, and $q_j = 1$ for each honeypot computers $j$. For a specific attacker, we measure its reconnaissance capability by comparing it with the perfect attacker, namely b comparing two sequences: the sequence of the expected {values} perceived by  an arbitrary attacker, $(a_j)_{j=1}^{n+m}$ where $a_j=(1-q_j) \cdot v_{j,A}$; the sequence of the expected {value} perceived by  the perfectly capable attacker,  $(v_{j,D})_{j=1}^{n+m}$. We measure the similarity between these two sequences by treating them as $(n+m)$-dimensional vectors and using the {\em cosine similarity} metric widely used in data science~\cite{thearling1999introduction}. The cosine similarity between vector $\Vec{A}=(a_j)_{j=1}^{n+m}$ and $\Vec{D}=(v_{j,D})_{j=1}^{n+m}$  
is defined as:
\vspace{-1em}
\begin{equation*}
     \mathcal{S}_C(\Vec{A},\Vec{D}) = \frac{\Vec{A} \cdot \Vec{D}}{\lVert \Vec{A} \rVert \lVert \Vec{D} \rVert}. 
\end{equation*}
If the cosine similarity is $0$, it means that the two vectors 
are orthogonal to each other in the sense that $a_j>0$ when $v_{j,D}= 0$ and $a_j= 0$ when $v_{j,D}>0$. In this case, the attacker is completely wrong, namely believing that production computers are honeypot computers and that honeypot computers are production computers. If the cosine similarity equals to $1$, then $\frac{a_j}{\sum_{j=1}^{n+m}a_j}= \frac{v_{j,D}}{\sum_{j=1}^{n+m}v_{j,D}}$ for all $j$, meaning the expected {value of each computer perceived by} the attacker is almost always proportional to $v_{j,D}$. 



In our experiment, we use different cosine similarity values by generating the $q_j$'s in a ``semi-random" fashion (because drawing $q_j$'s uniformly at random from $[0,1]$ always yields a large cosine similarity). More specifically, we generate the $q_j$'s according to normal distribution $\mathcal{N}(x,0.1)$ where $x\in\{0.1,0.25,0.5,0.75,0.9\}$, and for each $x$ we use $\mathcal{N}(x,0.1)$ to generate $20\%$ of the $q_j$'s. 

\vspace{-2em}
\begin{figure}[!htbp]
 \centering
	\includegraphics[width=.6\textwidth]{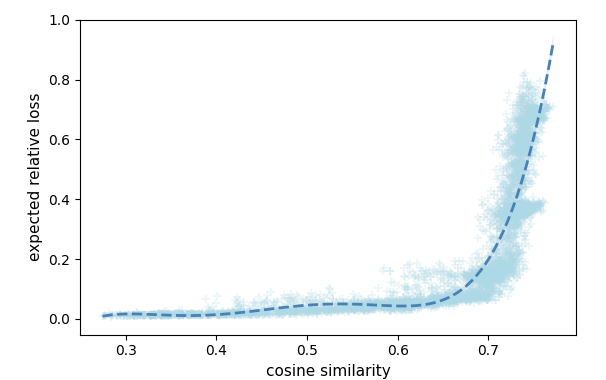}
 \vspace{-1em}
	\caption{Attacker's reconnaissance capability vs. the expected loss to the defender.}
	\label{fig:strength_correlation}
\end{figure}
\vspace{-2em}

Figure~\ref{fig:strength_correlation} plots the experimental result, showing that the expected relative loss increases marginally when the cosine similarity is below a certain threshold, but increases sharply when it is above a certain threshold. 
\begin{insight}
\label{insight-2}
Blending honeypots into production computers is extremely effective when the attacker's reconnaissance capability is below a threshold. 
\end{insight}

\section{Limitations}
\label{sec:limitations}

This study has a number of limitations. First, we assume that the attacker attacks computers in an independent fashion. In practice, the attacker may re-evaluate its perception of both $v_{j,A}$ and $q_j$ after attacking a computer. This is possible because the attacker will receive feedback from attacking a production computer that is different from attacking a honeypot computer. This poses an outstanding open problem for future research: How should we extend the model to incorporate this kind of feedback?
Second, we assume that the new attacks available to the attackers are equally capable of attacking any production computer and will be applicable to any honeypot computer. The former is a simplifying assumption because new attacks may have different capabilities 
and incur different costs (e.g., a zero-day exploit against an operating system would be more powerful and expensive than a zero-day exploit against an application program). The latter is also a simplifying assumption because, for example, an exploit against Microsoft Windows is not applicable to Linux.  Future research needs to investigate how to extend the model to accommodate such differences.
Third, we formalize the HD problem in a ``one-shot'' fashion, meaning that the honeypot computers, once deployed, are never re-deployed (i.e., their IP addresses never change after deployment). The effectiveness of honeypots would be improved by dynamically adjusting the locations of the honeypot computers.


\section{Related Work}
\label{sec:related-work}


\noindent{\bf Prior Studies Related to Honeypots}.
From a conceptual point of view, the present study follows the Cybersecurity Dynamics framework \cite{XuBookChapterCD2019,XuMTD2020,XuSciSec2021SARR}, which aims to rigorously and quantitatively model attack-defense interactions in cyberspace. From a technical point of view, honeypot is a  cyber deception technique.
We refer to \cite{zhu2021survey,al2019autonomous,wang2018cyber,DBLP:journals/csur/HanKB18,DBLP:books/sp/JSSW2016,rowe2016introduction} for  cyber deception in a broader context. 
We divide prior studies on honeypots into three families based on their purposes. 

The first purpose is to study how to leverage honeypots to detect new attacks 
(e.g., \cite{new.att.honeypot,honeypot.pca,li2010towards,DBLP:journals/ccr/KreibichC04,DBLP:journals/cn/PortokalidisB07,DBLP:conf/uss/AnagnostakisSAXMK05}). These studies typically assume that 
honeypot computers and production computers belong to two different networks; this isolation renders honeypot's utility 
questionable because it is easy for attackers to determine the presence of such honeypot networks or honeynets. The present study falls under this thrust of research but advocates blending honeypot computers into production computers.
Moreover, our study is through an innovative lens, which is to maximize the utility of honeypot in forcing 
attackers to expose their new attacks, while minimizing the loss to the defender in terms of its digital assets stored in the compromised production computers. To the best of our knowledge, this is the first study on modeling and analyzing the utility of honeypots.

The second purpose is to study how to prepare or use honeypots
\cite{4267549,carroll2011game,wang2020intelligent,DBLP:conf/hicss/MiahGVTK20,kulkarni2020decoy,anwar2021honeypot,wagener2011adaptive,DBLP:conf/gamesec/HuangZ20,DBLP:conf/gamesec/HuangZ19,aggarwal2021decoys,pibil2012game}.
For example, 
some studies are geared toward making honeypot computers and production computers look the same to disrupt attackers' reconnaissance process
\cite{aggarwal2021decoys,pibil2012game,DBLP:conf/hicss/MiahGVTK20,kulkarni2020decoy};
some studies are geared toward deploying honeypots to defend networks with known vulnerabilities 
(e.g., \cite{anwar2021honeypot}); some studies focus on making honeypot self-adaptive to attacks \cite{wagener2011adaptive,DBLP:conf/gamesec/HuangZ20,DBLP:conf/gamesec/HuangZ19}.
Our study is different from these studies for at least three reasons. (i) Putting into the terminology of our study, these studies can be understood as treating the $v_{j,A}$'s and the $q_j$'s as their goal of study.
Whereas, we treat the $v_{j,A}$'s and $q_j$'s as a stepping-stone for characterizing the utility of honeypots in forcing attackers to expose their new attacks, which are not known to the defender. This means that these studies, which lead to honeypot computers with various degrees of sophistication, can be incorporated into our model to formulate a more comprehensive framework.
(ii) These studies are dominated by game-theoretic models. By contrast, we use a combinatorial optimization approach. This difference in approach can be justified by the difference in the goals because we focus on characterizing the utility of honeypots in forcing attackers to expose their new attacks, while minimizing the loss to the defender in terms of digital assets. 
(iii) Some of these studies assume that the vulnerabilities are known (but unpatched). By contrast, we can accommodate both known (but unpatched) and unknown vulnerabilities (e.g., zero-day vulnerabilities unknown to the defender). 

The third purpose is to study how to leverage honeypot-collected data to forecast cyber threats \cite{XuIEEETIFS2013,XuIEEETIFS2015,XuPLoSOne2015,XuTechnometrics2017,XuMarkerPointProcess2017,DBLP:journals/ejisec/FangXXZ19,XuGrangerCausality2020}. 
These studies lead to innovative statistical or deep learning models which can accurately forecast the number or the type of incoming attacks.
However, these studies leverage traditional honeypot deployments mentioned above, namely that honeypot computers belong to a different network than the production network.
By contrast, we investigate how to optimally blend honeypot computers into production networks, which would enable of more realistic forecasting results \cite{XuSciSec2023-forecasting}.

\smallskip

\noindent{\bf Prior Studies Related to Our Hardness and Algorithmic Results}. We study the defender's optimization problem given a stochastic attacker. This problem is closely related to the bi-level optimization problem 
\cite{caprara2013complexity,dempe2000bilevel,chen2013approximation,qiu2015improved,DBLP:journals/tcs/PferschyNP19}. However, these results are all for a deterministic follower (attacker), while the HD problem studied in this paper involves a stochastic attacker who makes decisions in a probabilistic way. We are not aware of approximation algorithms for bi-level optimization problems where the follower (attacker) is stochastic.   

\section{Conclusion}
\label{sec:conclusion}
Honeypot is an important cyber defense technique, especially in forcing attackers to expose their new attacks (e.g., zero-day exploits).
However, the effectiveness of honeypots has not been systematically investigated. This motivated us to formalize the Honeypot Deployment (HD) problem as one manifestation of understanding the effectiveness of blending honeypot computers into production computers in an enterprise network. {We show that the HD problem is NP-hard, provide a polynomial time approximation scheme to solve it, and present experimental results to draw further insights.} The limitations mentioned above represent interesting open problems for future research. 


\noindent{\bf Acknowledgement}. 
This work was supported in part by NSF Grants \#2122631, \#2115134 and \#2004096, and Colorado State Bill 18-086.

%
%
%

\bibliographystyle{splncs04nat}
{\bibliography{honeypot_ref}}

\clearpage

\appendix

\section{Proof of Theorem~\ref{thm:np-h}}\label{appsec:hardness}
\begin{proof}[Proof of Theorem~\ref{thm:np-h}]
Given an arbitrary instance of Subset Product, we construct an instance of the HD problem as follows: there is only one production computer {(i.e., $n=1$) with value $1$}. There are $m+1$ IP addresses in total, where the production computer is the last computer (i.e., computer $m+1$). The defender can deploy honeypot computers at IP addresses from $1$ to $m$. In particular, deploying honeypot computer at IP address $i$ costs $c_i={\lfloor \log(w_i)M\rfloor}/{M}$. The defender's budget is $B=\lceil \log(k)M \rceil/M$ where $M=k(m+1)$. The probability is $q_i=1/w_i$ for non-dummy computer $i\in [m]$, and $q_{m+1}=0$.  
Set parameter $r=1$ and the threshold of the expected loss to the defender as $T=1/k$.

Note that although $\log k$ and $\log w_i$'s are not rational numbers, for the purpose of determining the value of (e.g.) $\lfloor \log (w_i)M\rfloor$, it suffices to compute $\log w_i$ up to a precision of $O(1/M)$, which can be done in $O(\log M)$ time. It is easy to verify that the input length of the HD problem is $O(m\log k+m\log w_{max})$ where $w_{max}=\max_iw_i$.

Consider the expected loss of the defender. Since the value of the production computer is $1$, $\mathbb{E}(Loss)$ equals the probability that computer $m+1$ is attacked. Since $r=1$, the attacker can attack computer $m+1$ only if it does not attack any computer from $1$ to $m$, which equals to 
$$\prod_{i\in [m], \text{a honeypot computer is deployed at } i}\frac{1}{w_i}.$$


Suppose the answer to the instance of the Subset Product instance is ``yes", then we know that there exists $S'$ such that $\prod_{i\in S'}w_i=k$. Deploying honeypot computers at IP addresses $i$ where $i\in S'$, we know the total cost to the defender equals 
$$\sum_{i\in S'}\frac{\lfloor M\log w_i\rfloor}{M}\le \sum_{i\in S'}\log w_i\le \log k \le \frac{\lceil M\log k\rceil}{M}=B,$$ which is within budget $B$. Meanwhile, the expected loss is
$$\prod_{i\in S'}\frac{1}{w_i}=1/k.$$
Hence, the minimum expected loss for the HD instance is no larger than $1/k$, i.e., the answer for the HD instance is ``yes''.

Suppose the answer for the HD instance is "yes" (i.e., the minimum expected loss to the defender is no larger than $1/k$). Let $S'\subseteq [m]$ be the IP addresses where the defender deploys honeypot computers, we know that
\begin{equation}
\sum_{i\in S'} \log(w_i) - m/M\leq\sum_{i\in S'}c_i\leq B \leq \log(k)+1/M,  \label{eq:hard-1}
\end{equation} and \begin{equation}
\prod_{i\in S'}1/w_i\leq 1/k. \label{eq:hard-2}
\end{equation}
By Combining Eq.\eqref{eq:hard-1} and Eq.\eqref{eq:hard-2}, we know that $$k\leq \prod_{i\in S'}w_i \leq k\cdot 2^{(m+1)/M} <k+1.$$
Consequently, $\prod_{i\in S'}w_i=k$.
Hence, the answer to the Subset Product instance is ``yes".
\end{proof}

\section{Proof of Lemma~\ref{lemma:dp}}\label{appsec:proof-lemmadp}

\begin{proof}
We prove this by induction. It is easy to verify Lemma~\ref{lemma:dp} for $t=1$. Suppose it holds for $t=\ell-1$, that is, there exists some $(\ell-1,p_{\ell-1,0}(x^*),\dots,p_{\ell-1,r}(x^*),$ $e_{\ell-1}x^*),b)\in \widetilde{\mathcal{F}}_{\ell-1}$ such that $b\leq b_{\ell-1}(x^*)$, we prove Lemma~\ref{lemma:dp} for $t=\ell$. We distinguish two cases based on computer $\ell$ and the optimal solution $x^*$: 
\begin{itemize}
\item  If $\ell\in\mathcal{P}$, then for the optimal solution $x^*$ it holds that $p_{\ell,k}(x^*)=p_{\ell-1,k}(x^*)$ for $0\le k\le r$, $e_\ell(x^*)=e_{\ell-1}(x^*)+(1-q_\ell)v_\ell\sum_{k=1}^{r-1}p_{\ell-1,k}$, and $b_{\ell}(x^*)=b_{\ell-1}(x^*)$. 
Given that $(\ell-1,p_{\ell-1,0}(x^*),\dots,p_{\ell-1,r}(x^*),e_{\ell-1}(x^*),b)\in \widetilde{\mathcal{F}}_{\ell-1}\subseteq \mathcal{F}_{\ell-1}$, according to the recursive computation of $\mathcal{F}_{\ell}$ from $\mathcal{F}_{\ell-1}$, $(\ell,p_{\ell,0}(x^*),$ $\dots,p_{\ell,r}(x^*),$ $e_\ell(x^*),b_\ell(x^*))\in {\mathcal{F}}_{\ell}$. Hence, there exists some  $(\ell,p_{\ell,0}(x^*),\dots,$ $p_{\ell,r}(x^*), e_\ell(x^*),b)\in \widetilde{\mathcal{F}}_{\ell}$ such that $b\leq b_\ell(x^*)$ by the definition of $\widetilde{\mathcal{F}}_{\ell}$.
\item If $\ell\in \mathcal{H}$, we further distinguish two sub-cases: (i) If computer $\ell$ is a dummy computer in $x^*$, then $p_{\ell,k}(x^*)=p_{\ell-1,k}(x^*)$ for $0\le k\le r$, $e_\ell(x^*)=e_{\ell-1}(x^*)$, and $b_{\ell}(x^*)=b_{\ell-1}(x^*)$. (ii) If computer $\ell$ is a honeypot computer in $x^*$, then we have $p_{\ell,0}(x^*)=q_\ell p_{\ell-1,0}(x^*)$, $p_{\ell,k}(x^*)=p_{\ell-1,k}(x^*)q_\ell+(1-q_\ell)p_{\ell-1,k-1}(x^*)$ for $1\le k\le r$, 
and $e_\ell(x^*)=e_{\ell-1}(x^*)$, $b_{\ell}(x^*)=b_{\ell-1}(x^*)+c_t$. In both sub-cases, the recursive computation of $\mathcal{F}_{\ell}$ from $\mathcal{F}_{\ell-1}$ implies that $(\ell,p_{\ell,0}(x^*),\dots,$ $p_{\ell,r}(x^*)$, $e_\ell(x^*),b_{\ell}(x^*))\in {\mathcal{F}}_{\ell}$. Hence there exists  $(\ell,p_{\ell,0}(x^*),\dots,p_{\ell,r}(x^*)$, $e_\ell(x^*),b)\in \widetilde{\mathcal{F}}_{\ell}$ such that $b\leq b_\ell(x^*)$ by the definition of $\widetilde{\mathcal{F}}_{\ell}$.
\end{itemize}
Hence, Lemma~\ref{lemma:dp} is true for all $t\in [n+m]$.

\end{proof}
\end{document}